\documentclass[leqno,11pt]{amsart}
\usepackage{amsmath,amssymb,amsfonts,amsthm}
\usepackage{amsfonts}
\usepackage{graphics}
\usepackage{graphicx}
\usepackage{subfigure}

\newtheorem{theorem}{Theorem}

\newtheorem{example}[theorem]{Example}

\newtheorem{lemma}[theorem]{Lemma}

\newtheorem{proposition}[theorem]{Proposition}
\newtheorem{remark}[theorem]{Remark}

\begin{document}
\title[]{ Pricing Temperature Derivatives under a Time-Changed Levy Model}
\author[]{Pablo Olivares, Ryerson University}
\address{}

\begin{abstract}
The objective of the paper is to  price weather contracts using temperature  as the underlying process when the later follows a mean-reverting dynamics driven by a time-changed Brownian motion coupled to a Gamma Levy subordinator and time-dependent deterministic volatility. This type of model  captures the complexity of the temperature dynamic providing a more accurate valuation of their associate weather contracts. An approximated price is obtained by a Fourier expansion of its characteristic function combined with a selection of the equivalent martingale measure following the Esscher transform proposed in Gerber and Shiu (1994).
\end{abstract}

\maketitle
\section{Introduction:}
The objective of the paper is to  price weather contracts using temperature  as the underlying process when the later follows a mean-reverting dynamics driven by a time-changed Brownian motion coupled to a Gamma Levy subordinator and a time-dependent volatility function. The process reverts to a seasonal periodic deterministic process, while the volatility is considered also a periodic function of time, see Dacunha-Castelle, Hoang and  Parey (2015) for the later. Temperature models driven by Levy noises and stochastic volatility have been originally considered in  Benth and Benth-S(2009). \\
This type of model  captures the complexity of the temperature dynamic providing a more accurate valuation of their associate weather contracts.\\
On the other hand, the availability of an explicit analytical expression of the characteristic function of the process allows for its Fourier expansion with respect of its characteristic function, which in turn  leads to compute the approximated price under an equivalent martingale measure (EMM) obtained from the Esscher transform, see Gerber and Shiu (1994). \\
The combination of these three elements, namely the model, the pricing method and the choice of the EMM in the context of weather derivatives offers a novel methodology for pricing such contracts. \\
 Methods based on Fourier expansions of the characteristic function in one and two dimensions are implement in Fang and Oosterlee (2008)to European contracts and further extended to other derivatives by the same authors, see Fang and Oosterlee (2014).\\
  Finally, we fit the model to a series of  daily average temperatures at  Pearson airport, Ontario, Canada during the period 2014-2019.\\
The organization of the paper is the following:\\
In section 2 we describe the main model for the temperature process and obtain the characteristic function associated with it. In section 3 we discuss  the implementation of the Fourier expansion techniques, while in section 4 we show the numerical results in the fitting of the model,  pricing results and their sensitivities to key parameters.
\section{Modeling temperature }
 Let  $(\Omega ,\mathcal{A}, (\mathcal{F}_{t})_{t \geq 0}, P)$ be a filtered probability space verifying the usual conditions. For a stochastic process $(X_t)_{t \geq 0}$ defined on the space filtered space above the functions $\varphi_{V_{t}}$  and $l_X(u)=\frac{1}{t} \log \varphi_{V_{t}}(-iu)$ defines its characteristic function and the cumulat generating function respectively. When the process has stationary and independent increments the later does not depend on $t$. The $\sigma$-algebra $\mathcal{F}_{Y_t}=\sigma(Y_u, 0 \leq u \leq t)$ is the $\sigma$-algebra generated by the random variables $Y_u, 0 \leq u \leq t$. The changes of the temperature over an interval $[t,t+h)$ are denoted $\Delta T_t=T_{t+h}-T_t$. For a process $(X_t)_{t \geq 0}$, the discounted process $(\tilde{X}_t)_{t \geq 0}$ is defined as $\tilde{X}_t=e^{-rt} X_t$, where $r$ is the contstant interest rate.\\
Let $(T_t)_{t \geq 0}$ be the daily average temperature process defined on the filtered space above. The average temperature is taken as the arithmetic mean between the maximum and the minimum temperature during a given day.\\
 We assume the temperature process $(T_t)_{t \geq 0}$ verifies the stochastic differential equation:
\begin{equation}\label{eq:tempdyna}
    dT_t= \alpha(s_t-T_t)dt+  \sigma_t dV_t
\end{equation}
where  $(s_t)_{t \geq 0}$ is a deterministic seasonal process such that:
\begin{eqnarray} \label{eq:seass2}
s_t&=&  \beta_0 + \beta_1 t+ \beta_2 \sin \left(\frac{2 \pi}{365} t \right)+\beta_3 \cos \left(\frac{2 \pi}{365}t \right)
\end{eqnarray}
The parameter $\alpha$ is the mean-reversion rate to the seasonal component. The background noise $(V_t)_{t \geq 0}$ will be specified later on.\\
The solution of equation (\ref{eq:tempdyna}) is given in the following lemma.
\begin{lemma}
The solution of equation (\ref{eq:tempdyna}) is:
\begin{eqnarray} \label{eq:timechansol}
  T_t&=&e^{-\alpha t} T_0+ \alpha K_1(t, \alpha)+ W_t
  \end{eqnarray}
with $W_t=\int_0^t  \sigma_u e^{-\alpha (t-u)}dV_u$ and
\begin{eqnarray*}
 K_1(t, \alpha) &=& \int_0^{t} s_u e^{-\alpha(t-u)}\;du= \frac{1}{\alpha}(1-e^{-\alpha t})\beta_0+ \frac{1}{\alpha}(1-\frac{1}{\alpha})(1-e^{-\alpha t})\beta_1\\
 &+& \frac{1}{\alpha}\frac{[\cos(\frac{2 \pi}{365}t)-e^{-\alpha t}-\frac{1}{\alpha}\frac{2 \pi}{365}\sin(\frac{2 \pi}{365}t)]}{1-\frac{1}{\alpha^2}(\frac{2 \pi}{365})^2}\beta_2\\
 &+& \frac{1}{\alpha}\frac{[\sin(\frac{2 \pi}{365}t)+\frac{1}{\alpha}\frac{2 \pi}{365}(e^{-\alpha t}-\cos(\frac{2 \pi}{365}t)]}{1+\frac{1}{\alpha^2}(\frac{2 \pi}{365})^2}\beta_3
\end{eqnarray*}
\end{lemma}
\begin{proof}
We apply Ito formula to the function $f(x,y)=x e^{\alpha y}$ and the process $(T_t,t)$.\\
 Hence:
\begin{eqnarray*}
 T_t e^{\alpha t} &=& T_0+\int_0^t e^{\alpha u}dT_{u^-}+ \alpha \int_0^t e^{\alpha u}T_{u^-} du\\
 &+& \sum_{u \leq t} [T_{u}e^{\alpha  u}- T_{u^-} e^{\alpha  u}-\Delta T_{u^-}e^{\alpha u}]\\
  &=& T_0+ \alpha \int_0^t ( s_u-T_{u^-})e^{\alpha u}du+ \int_0^t  \sigma_u e^{\alpha u}dV_u+\alpha \int_0^t T_{u^-} e^{\alpha u}du\\
    &=& T_0+ \alpha e^{\alpha t}K(t, \alpha)+ \int_0^t  \sigma_u e^{\alpha u}dV_u
    \end{eqnarray*}
 Multiplying by $e^{-\alpha t}$ on both sides leads to equation (\ref{eq:timechansol}).
\end{proof}
We assume the volatility also follows a deterministic seasonal component process:
\begin{equation}\label{eq:voldet}
    \sigma_t= c_0 + c_1 t+ c_2 \sin \left(\frac{2 \pi}{365} t \right)+c_3 \cos \left(\frac{2 \pi}{365}t \right)
   \end{equation}
where $c_j \geq 0, j=0,1,2,3.$\\
We will need to compute the characteristic function of some integrals of the background noise process. To this end we will make use of a well-known result  about functional  of a Levy process $(\xi_t)_{t \geq 0}$ and a measurable function $f$:
\begin{equation}\label{eq:funclevy}
  E ( exp( i \int_0^t f(s)\;d \xi_s) )=exp(\int_0^t l_{\xi}(-i f(s))\;ds)
\end{equation}
In order to select the EMM for pricing purposes we take an Esscher transform of the historic measure $P$. See Gerber and Shiu(1994) for a rationale in terms of a utility-maximization criteria. \\
For a stochastic process $(X_t)_{t \geq 0}$ we consider its Esscher transform:
  \begin{equation}\label{eq:esscher}
  \frac{d \mathcal{Q}^{\theta}_t}{d P_t}=\exp(\theta X_t-t l_X(\theta)),\; 0 \leq t \leq T,\; \theta \in \mathbb{R}
\end{equation}
   where $P_t$ and $\mathcal{Q}^{\theta}_t$ are the respective restrictions of $P$ and $\mathcal{Q}^{\theta}$ to the $\sigma$-algebra $\mathcal{F}_t$. We define by $\varphi^{\theta}_{X_t}$ and $ l^{\theta}_X(u)$  respectively the characteristic function and  moment generating function  of a process $(X_t)_{t \geq 0}$ under the probability $\mathcal{Q}^{\theta}$ obtained by an Esscher transformation as given in equation (\ref{eq:esscher}).\\
      For consistency we denote $\varphi^{0}_{X_{t}}:=\varphi_{X_t}$ and $ l_X^{0}=l_V$.\\
      By analogy with the case of financial underlying assets the risk market premium measure $\mathcal{Q}^{\theta}$ making the discounted temperatures process $(\tilde{T}_t)_{t \geq 0}$ a martingale  for $r>0$ a fixed interest rate is called an Equivalent Martingale Measure (EMM). The expected value under $\mathcal{Q}^{\theta}$ is denoted $E_{\theta}$.\\
  We set a subordinator process $(R_t)_{t \geq 0}$ and the time-changed process $(V_t)_{t \geq 0}$ verifying:
\begin{eqnarray}\label{eq:tchlevy2}
  V_t&=& B_{R_t}+ \mu_1 R_t
\end{eqnarray}
Here  $\mu_1   \in \mathbb{R}$  is a  parameters in the model and $(B_t)_{t \geq 0}$ is a standard Brownian motion.\\
The following result describes the characteristic function of the temperature process under the historic measure $P$.
\begin{proposition}\label{prop:chftempthetsdet}
Under the model described by equations (\ref{eq:tempdyna}), (\ref{eq:seass2}) and (\ref{eq:tchlevy2})   the  characteristic function of $T_t$ under the probability $P$ is:
\begin{eqnarray}\label{eq:chftdet}
 \varphi_{T_t}(u) &=&  C_1(t, \alpha) \exp( \int_0^t  l_R(-i u \mu_1 \sigma_s e^{-\alpha (t-s)}-\frac{1}{2}  u^2 \sigma^2_s e^{-2 \alpha (t-s)})ds
\end{eqnarray}
 where:
\begin{eqnarray*}
     C_1(t, \alpha) &=&exp(iu e^{-\alpha t T_0}+\alpha K_1(t, \alpha))
\end{eqnarray*}
\end{proposition}
\begin{proof}
By conditioning:
\begin{eqnarray*}\nonumber
  \varphi_{V_t}(u)&=& E[E[\exp(i (u V_t)/R_t)] ]=E[ \exp(i u \mu_1 R_t) E[\exp( iu B_{R_t}/R_t)] ]\\
  &=& E[ \exp(i u \mu_1 R_t) \exp( -\frac{1}{2} R_t u^2)]= E[ \exp(i( u \mu_1+\frac{1}{2}i  u^2)R_t)]\\
  &=&  \varphi_{R_t}(u \mu_1+\frac{1}{2}i  u^2)
\end{eqnarray*}
Hence:
\begin{equation}\label{eq:lv}
  l_V(u)=l_R(u \mu_1+\frac{1}{2}  u^2)
\end{equation}
By lemma 1 and  formula (\ref{eq:funclevy}):
 \begin{eqnarray}\nonumber
  \varphi_{T_t}(u)&=& E[e^{iu T_t}]=C_1(t, \alpha) E[exp(iu \int_0^t  \sigma_s e^{-\alpha (t-s)}dV_s)]\\
  &=& C_1(t, \alpha) exp( \int_0^t  l_V(-i u \sigma_s e^{-\alpha (t-s)})ds)
\end{eqnarray}
Combined with equation (\ref{eq:lv}), equation (\ref{eq:chftdet}) immediately follows.
\end{proof}
 The results below provides the characteristic function of the temperature process  under the EMM defined via an Esscher transform.
   \begin{proposition}
 Let $(T_t)_{t \geq 0}$ be the temperature process defined by equations (\ref{eq:tempdyna})-(\ref{eq:tchlevy2}). Then, the characteristic function  under the Esscher EMM $\mathcal{Q}^{\theta}$ is:
    \begin{eqnarray}\nonumber
        \varphi^{\theta}_{T_t}(u)&=& C_1(t, \alpha) C_2(t, \theta) I_t(u,\theta)\\ \label{eq:chfuntemp}
        &&
      \end{eqnarray}
where $C_1(t, \alpha)$ is defined as in the previous proposition and:
    \begin{eqnarray*}
     C_2(t, \theta) &=& \exp(-t l_R(\theta \mu_1+\frac{1}{2}  \theta^2))\\
     I_t(u,\theta)&=& \exp(\int_0^t l_R(-i u \mu_1 \sigma_s e^{-\alpha (t-s)}+ \mu_1\theta +\frac{1}{2}  (-i u \sigma_s e^{ \alpha (t-s)}+\theta)^2)ds)
  \end{eqnarray*}
  and for any $T>0$ the parameter $\theta$ verifies:
   \begin{equation}\label{eq:gershui}
l'_{V}(\theta )=-e^{(\alpha+r)T}(1-\tilde{C}_2(T, \alpha)) K^{-1}_2(\alpha,T)
   \end{equation}
   where:
   \begin{eqnarray*}
   K_2(\alpha,T) &=&  \int_0^T \sigma_u e^{\alpha u}\;du =  \frac{c_0}{\alpha}(e^{\alpha T}-1)+\frac{c_1 T}{\alpha}e^{\alpha T}-\frac{c_1}{\alpha^2}(e^{\alpha T}-1)\\
 &-& \frac{365}{2 \pi}c_2 \left(\cos(\frac{2 \pi}{365 T})-1 \right)+  \frac{365}{2 \pi}c_3 \left(\sin(\frac{2 \pi}{365 T})-1 \right)
   \end{eqnarray*}
   \end{proposition}
    \begin{proof}
    Notice that:
   \begin{eqnarray*}
    \varphi^{\theta}_{V_t}(u) &=&  E_{\theta}(e^{i u V_t} e^{\theta V_t-t l_V(\theta)})=\frac{\varphi_{V_t}(u- i \theta)}{\varphi_{V_t}(- i \theta)}\\
   \end{eqnarray*}
   and $ l^{\theta}_{V}(u) = l_V(u+\theta)-l_V(\theta)$.\\
  Then, similarly to proposition \ref{prop:chftempthetsdet}:
    \begin{eqnarray*}
    \varphi^{\theta}_{T_t}(u) &=& C_1(t, \alpha) \exp(\int_0^t l_V^{\theta}(-i u \sigma_s e^{-\alpha (t-s)})ds)\\
    &=& C_1(t, \alpha) \exp(-t l_V(\theta)) \exp(\int_0^t l_V(-i u \sigma_s e^{-\alpha (t-s)}+\theta)ds)
     \end{eqnarray*}
     from which equation (\ref{eq:chfuntemp}) follows.\\
     By equation (\ref{eq:esscher}) the discounted temperature process $(\tilde{T}_t)_{t \geq 0}$ verifies:
     \begin{equation*}
       \tilde{T}_t= \tilde{C}_2(t, \alpha)+\tilde{W}_t
     \end{equation*}
    It is a $\mathcal{Q}^{\theta}$-martingale if and only if for any $0 \leq s < t $:
   \begin{eqnarray*}
   E_{\theta}(\tilde{T}_t/ \mathcal{F}_s)  &=& \tilde{T}_s\\
    \Leftrightarrow   E_{\theta}(\tilde{W}_t-\tilde{W}_s/ \mathcal{F}_s)  &=&  \tilde{C}_2(s, \alpha)-\tilde{C}_2(t, \alpha)
   \end{eqnarray*}
   But:
   \begin{eqnarray*}
    E_{\theta}(\tilde{W}_t-\tilde{W}_s/ \mathcal{F}_s) &=&  E_{\theta}(e^{-( \alpha+r) t}\int_0^t \sigma_u e^{\alpha u}\;dV_u-e^{-(\alpha+r) s}\int_0^s \sigma_u e^{\alpha u}\;dV_u/ \mathcal{F}_s)\\
    &=&  E_{\theta}(e^{-(\alpha+r) t}\int_s^t \sigma_u e^{\alpha u}\;dV_u+(e^{-( \alpha+r) t}-e^{- (\alpha+r) s})\int_0^s \sigma_u e^{\alpha u}\;dV_u/ \mathcal{F}_s)\\
    &=&  E_{\theta}(e^{-(\alpha+r) t}\int_s^t \sigma_u e^{\alpha u}\;dV_u)\\
    &+& (e^{-( \alpha+r) t}-e^{-( \alpha+r)s})\int_0^s \sigma_u e^{\alpha u}\;dV_u\\
       \end{eqnarray*}
       On the other hand, from equation (\ref{eq:funclevy}):
       \begin{eqnarray}\nonumber
  \varphi^{\theta}_{V_t}(x)&=& exp(\int_0^t l^{\theta}_{V}(i x \sigma_u e^{\alpha u} )\;du)=exp(\int_0^t( l_{V}(i x \sigma_u e^{\alpha u}+\theta )-l_{V}(\theta)) \;du)\\ \label{eq:funclevy2}
  &&
  \end{eqnarray}
  Hence:
  \begin{eqnarray*}
   E_{\theta}(e^{-(\alpha+r) t}\int_s^t \sigma_u e^{\alpha u}\;dV_u) &=& e^{-(\alpha+r) t} \frac{1}{i}(\varphi^{\theta}_{V_t})'(x)|_{x=0} \\
    &=& -e^{-(\alpha+r) t} \frac{1}{i}(i\int_s^t \sigma_u e^{\alpha u}l'_{V}(-i x \sigma_u e^{\alpha u}+\theta )\;du|_{x=0}\\
   && \exp(\int_s^t (l_{V}(-i x \sigma_u e^{\alpha u}+\theta )-l_{V}(\theta)) \;du)|_{x=0}\\
   &=& -e^{-(\alpha+r) t} l'_{V}(\theta ) \int_s^t \sigma_u e^{\alpha u}\;du
  \end{eqnarray*}
   In particular for $t=T$ and $u=0$ we have the result in equation (\ref{eq:gershui}), that follows from elementary calculation.
        \end{proof}
      \begin{remark}
   Notice that the characteristic function under the probability $P$ is obtained from equation (\ref{eq:chfuntemp}) taking $\theta=0$. Hence we write $I_t(u)=I_t(u,0)$, $\varphi^{0}_{Y_t}=\varphi_{Y_t}$ and $\mathcal{Q}^{0}=P$.
   \end{remark}
   \begin{example}\textit{Gamma subordinator}\\
   Consider the subordinator $(R_t)_{t \geq 0}$  is a Gamma process with parameters $a >0,b>0$, see Carr  and Madan (1999), with respective
characteristic function and Laplace exponent:
\begin{eqnarray*}
 \varphi_{R_t}(u)  &=&  \left( 1- \frac{i u}{b} \right)^{-a t},\;a>0, b>0\\
 l_{R}(u)  &=&  -a \log \left( 1- \frac{ u}{b}\right),\; u<b
\end{eqnarray*}
Therefore:
\begin{eqnarray*}
 \varphi_{V_t}(u)&=& \varphi_{R_t}(\mu_1 u +\frac{1}{2}i u^2)=\left( 1-i \frac{(\mu_1 u +\frac{1}{2}i u^2)}{b} \right)^{-at} \\
 &=& \left( 1- \frac{i\mu_1 u }{b}+\frac{1}{2b}u^2 \right)^{-at}\\
 l_V(u) &=& -a \log A_1(u)
  \end{eqnarray*}
  where:
   \begin{equation*}
     A_1(u)= 1- \frac{\mu_1 u }{b}-\frac{1}{2b}u^2
   \end{equation*}
  Moreover:
  \begin{eqnarray*}
 l^{\theta}_V(u) &=& l_V(u+\theta)-l_V(\theta)\\
 &=& -a \left[\log A_1(u+\theta) - \log  A_1(\theta)  \right]\\
 &=& -a \log  \left( \frac{A_1(u+\theta)}{A_1(\theta)} \right) \\
 &=& -a \log \left( 1-\frac{ \frac{\mu_1 u }{b}-\frac{1}{2b}(u^2+2 \theta u) }{ 1- \frac{\mu_1 \theta }{b}-\frac{1}{2b}\theta^2} \right)
  \end{eqnarray*}
 To compute the characteristic function of the temperature $T_t$ under the EMM Esscher transformation given by equation (\ref{eq:chfuntemp}) we have:
 \begin{eqnarray*}
  C_1(t, \alpha) &=&exp(iu e^{-\alpha t T_0}+\alpha K_1(t, \alpha))\\
   C_2(t, \theta) &=& \exp(-t l_R(\theta \mu_1+\frac{1}{2}  \theta^2))=A^{at}_1(\theta)\\
   I_t(u,\theta)&=& \exp(\int_0^t l^{\theta}_V(-iu \sigma_s e^{-\alpha(t-s)})ds)\\
    &=&  \exp \left(-a \int_0^t  \log \left( \frac{A_1(-iu \sigma_s e^{-\alpha(t-s)}+ \theta)  }{ A_1(\theta)}\right)\;ds   \right)
       \end{eqnarray*}
  To compute the  Gerber-Shiu  parameter, from the martingale condition given by equation (\ref{eq:gershui}):
\begin{eqnarray*}
l'_{V}(\theta )&=&  \frac{a(\mu_1+\frac{1}{2}\theta)}{b  A_1(\theta)}=-e^{(\alpha+r)T}(1-\tilde{C}_2(T, \alpha)) K_2(\alpha,T)^{-1})\\
&=& -e^{(\alpha+r)T}(1- e^{-rT} A^{aT}_1(\theta) K_2(\alpha,T)^{-1})\\
&=& -e^{(\alpha+r)T}+ e^{\alpha T} A^{aT}_1(\theta) K_2(\alpha,T)^{-1}\\
\end{eqnarray*}
Therefore, the value $\theta^*$ that solves:
\begin{equation}\label{eq:gershuigamma}
\mu_1+\frac{1}{2}\theta +\frac{b}{a}e^{(\alpha+r)T} A_1(\theta)- \frac{b}{a} e^{\alpha T} A^{aT+1}_1(\theta) K_2(\alpha,T)^{-1}=0
   \end{equation}
 makes the discounted prices martingales under the Esscher transformation.
         \end{example}
\section{Pricing weather options}
Weather contracts are based on  cumulate temperatures (CAT), heating-degrees-days (HDD) or cooling-degrees-days (CDD) over certain period $[0,T]$. Futures and option  contracts are offered in Chicago Mercantile Exchange. They are respectively defined as:
\begin{eqnarray*}\label{eq:cat}
 \xi_T&=& CAT = \sum_{k=1}^{T} T_k\\ \label{eq:hdd}
 \xi_{2,T}&=& HDD =  \sum_{k=1}^{T} (c-T_k)_+ \\ \label{eq:cdd}
 \xi_{3,T}&=& CDD = \sum_{k=1}^{T} (T_k-c)_+
\end{eqnarray*}
The typical case is $c=18^o$ Celsius.\\
For concreteness we focus on a CAT index. To this end for convenience we rewrite the CAT index as:
\begin{eqnarray}\nonumber
 \xi_T &=&  \sum_{k=1}^{T} T_k= \sum_{t=1}^{T} (T_0+\sum_{j=1}^{t} \Delta T_j)\\ \label{eq:cat2}
  &=& T T_0+\sum_{j=1}^{T} \gamma_j \Delta T_{j}
\end{eqnarray}
where the changes in temperature $\Delta T_{j}=T_{j+1}-T_j$  are independent random variables and $\gamma_j=T-j+1$.\\
 A general payoff of the temperature weather derivative, consisting in a combination of a European long put and a long call with different strikes, known as \textit{strangle}, is given by:
 \begin{equation}\label{eq:payoff}
  h(\xi_T)=d_1(\xi_T-K_1)_++d_2(K_2-\xi_T)_+ , d_j>0, K_1>K_2>0 \;,j=1,2
\end{equation}
where $d_1$ and $d_2$ are the costs per unit  of temperature below (resp. above) the threshold $K_1$ (resp. $K_2$) known as \textit{tick sizes}.\\
The price of a temperature contract over the period $[0, \frac{T}{365}]$  is :
\begin{eqnarray}\nonumber
 p_W &=& d_1  e^{-r \frac{T}{365}} E_{\mathcal{Q}}(\xi_T-K_1)_++d_2 e^{-r \frac{T}{365}} E_{\mathcal{Q}}(K_2-\xi_T)_+\\ \nonumber
&=& d_1  e^{-r \frac{T}{365}} \int_{\mathbb{R}}(x-K_1)_+f_{\xi_T}(x, \theta)\;dx + d_2 e^{-r \frac{T}{365}} \int_{\mathbb{R}}(K_2-x)_+f_{\xi_T}(x, \theta)\;dx \\ \label{eq:price}
&&
\end{eqnarray}
where $r$ is the interest rate and $f_{\xi_T}(x, \theta)$ is the p.d.f. of the cumulated temperature under the EMM measure.
A Fourier expansion of the p.d.f. $f_{\xi_T}(x, \theta)$ on an  interval $[b_1,b_2]$ is given by:
\begin{eqnarray}\label{eq:pdfcumtemp}
  f_{\xi_{1}}(x, \theta) &=&  \sum_{k=0}^{+\infty} A_{k}(\theta) cos \left( k \pi \frac{x-b_1}{b_2-b_1}\right)
\end{eqnarray}
where the coefficients in the expansion, the first of them divided by two, are:
 \begin{eqnarray}\nonumber
  A_k(\theta) &=& \frac{2}{b_2-b_1} \int_{b_1}^{b_2} f_{\xi_T}(y, \theta) cos \Big(k \pi \frac{y-b_1}{b_2-b_1} \Big)\;dy\\ \nonumber
  & \simeq & \frac{2}{b_2-b_1} \int_{b_1}^{b_2} f_{\xi_T}(y, \theta) Re \Big(e^{i k \pi \frac{y-b_1}{b_2-b_1}} \Big)\;dy\\ \nonumber
  &=& \frac{2}{b_2-b_1} Re \Big( \int_{b_1}^{b_2} f_{\xi_T}(y, \theta)e^{i k \pi \frac{y-b_1}{b_2-b_1}} \;dy \Big)\\ \nonumber
  &=& \frac{2}{b_2-b_1} exp \left(-i \frac{ k \pi b_1}{b_2-b_1} \right) \varphi^{\theta}_{\xi_T} \Big(\frac{k \pi}{b_2-b_1} \Big) \\ \label{eq:coeff}
  &&
  \end{eqnarray}
Replacing (\ref{eq:coeff}) into (\ref{eq:pdfcumtemp}), then (\ref{eq:pdfcumtemp}) in (\ref{eq:price}) we have:
\begin{eqnarray*}
&& \int_{\mathbb{R}}(x-K_1)_+f_{\xi_T}(x, \theta)\;dx  \simeq  \sum_{k=0}^{+\infty} A_{k}(\theta) \int_{b_1}^{b_2}(x-K_1)_+ cos \left( k \pi \frac{x-b_1}{b_2-b_1}\right) \;dx\\
& \simeq & \sum_{k=0}^{N_1} A_{k}(\theta) \int_{b_3}^{b_2}(x-K_1) cos \left( k \pi \frac{x-b_1}{b_2-b_1}\right) \;dx\\
& = & \sum_{k=0}^{N_1} A_{k}(\theta) \int_{b_3}^{b_2} x cos \left( k \pi \frac{x-b_1}{b_2-b_1}\right) \;dx - K_1 \int_{b_3}^{b_2} cos \left( k \pi \frac{x-b_1}{b_2-b_1}\right) \;dx \\
& = & \frac{2}{b_2-b_1} \sum_{k=0}^{N_1} exp \left(-i \frac{ k \pi b_1}{b_2-b_1} \right) \varphi^{\theta}_{\xi_T} \Big(\frac{k \pi}{b_2-b_1} \Big) \int_{b_3}^{b_2}x cos \left( k \pi \frac{x-b_1}{b_2-b_1}\right) \;dx \\
&-& \frac{2(K_1}{b_2-b_1} \sum_{k=0}^{N_1} exp \left(-i \frac{ k \pi b_1}{b_2-b_1} \right) \varphi^{\theta}_{\xi_T} \Big(\frac{k \pi}{b_2-b_1} \Big) \int_{b_3}^{b_2} cos \left( k \pi \frac{x-b_1}{b_2-b_1}\right) \;dx \\
& = & \frac{2}{b_2-b_1} \sum_{k=0}^{N_1} exp \left(-i \frac{ k \pi b_1}{b_2-b_1} \right) \varphi^{\theta}_{\xi_T} \Big(\frac{k \pi}{b_2-b_1} \Big) \int_{b_3}^{b_2}x cos \left( k \pi \frac{x-b_1}{b_2-b_1}\right) \;dx \\
&-& \frac{2 K_1}{b_2-b_1} \sum_{k=0}^{N_1} exp \left(-i \frac{ k \pi b_1}{b_2-b_1} \right) \varphi^{\theta}_{\xi_T} \Big(\frac{k \pi}{b_2-b_1} \Big) \int_{b_3}^{b_2}  cos \left( k \pi \frac{x-b_1}{b_2-b_1}\right) \;dx
\end{eqnarray*}
where $b_3=max(b_1, K_1) < b_2$ and from equation (\ref{eq:cat2}):
\begin{equation*}
  \varphi^{\theta}_{\xi_T}(u)=e^{i T T_0}\prod_{j=1}^T \varphi^{\theta}_{\Delta T_j}(\gamma_j u)
\end{equation*}
Moreover, for $k >0$:
\begin{eqnarray*}
&& \int_{b_3}^{b_2}x cos \left( k \pi \frac{x-b_1}{b_2-b_1}\right) \;dx = \frac{(b_1-b_2)b_3}{k \pi}\sin \left( k \pi \frac{b_3-b_1}{b_2-b_1} \right)\\
&+& \left( \frac{(b_2-b_1)}{k \pi} \right)^2 \left((-1)^k-\cos \left( k \pi \frac{b_3-b_1}{b_2-b_1} \right) \right)\\
&& \int_{b_3}^{b_2}  cos \left( k \pi \frac{x-b_1}{b_2-b_1}\right) \;dx = \frac{b_1-b_2}{k \pi} \sin \left( k \pi \frac{b_3-b_1}{b_2-b_1}\right)
\end{eqnarray*}
Then, separating the first term in the summation:
\begin{eqnarray*}
&& \int_{\mathbb{R}}(x-K_1)_+f_{\xi_T}(x, \theta)\;dx  \simeq   \frac{(b_2-b_3)^2}{2(b_2-b_1)} \\
 &+&  \frac{2}{b_2-b_1}  \sum_{k=1}^{N_1} exp \left(-i \frac{ k \pi b_1}{b_2-b_1} \right) \varphi^{\theta}_{\xi_T} \Big(\frac{k \pi}{b_2-b_1} \Big)\\
&& \left( \frac{(b_1-b_2)b_3}{k \pi}\sin \left( k \pi \frac{b_3-b_1}{b_2-b_1} \right)+ \left( \frac{(b_2-b_1)}{k \pi} \right)^2 \left((-1)^k-\cos \left( k \pi \frac{b_3-b_1}{b_2-b_1} \right) \right) \right)\\
&+& 2 K_1  \sum_{k=1}^{N_1} \frac{1}{k \pi} \exp \left(-i \frac{ k \pi b_1}{b_2-b_1} \right) \varphi^{\theta}_{\xi_T} \Big(\frac{k \pi}{b_2-b_1} \Big) \sin \left( k \pi \frac{b_3-b_1}{b_2-b_1}\right)\\
\end{eqnarray*}
In a similar analysis:
\begin{eqnarray*}
&& \int_{\mathbb{R}}(K_2-x)_+f_{\xi_T}(x, \theta)\;dx  \simeq  \frac{K_2}{2} [b_4-b_1-\frac{(b_4-b_1)^2}{2}]\\
  &-&  2 K_2  \sum_{k=1}^{N_2} \frac{1}{k \pi} exp \left(-i \frac{ k \pi b_1}{b_2-b_1} \right) \varphi^{\theta}_{\xi_T} \Big(\frac{k \pi}{b_2-b_1} \Big) \sin \left( k \pi \frac{b_4-b_1}{b_2-b_1}\right)\\
  &-& b_4 \sum_{k=1}^{N_2} \frac{1}{k \pi} exp \left(-i \frac{ k \pi b_1}{b_2-b_1} \right) \varphi^{\theta}_{\xi_T} \Big(\frac{k \pi}{b_2-b_1} \Big)\\
  && \left( \frac{(b_2-b_1)b_4}{k \pi}\sin \left( k \pi \frac{b_4-b_1}{b_2-b_1} \right)+ \left( \frac{(b_2-b_1)}{k \pi} \right)^2 \left(1-\cos \left( k \pi \frac{b_4-b_1}{b_2-b_1} \right) \right) \right)
\end{eqnarray*}
where $b_4=min(b_2,K_2)$.
The delicate choice of the truncation values $b_1$ and $b_2$ as well as the number of terms in the truncated expansion depends on the model considered, it is discussed in Fang and Oosterlee (2008). For detailed error analysis of the truncation and numerical errors present in the Fourier Cosine method we refer the reader to the work of Fang and Oosterlee (2008). We address this issue in the next section related to numerical aspects of the method.
\section{Numerical results}
We divide the section into three parts. In the first one we do a descriptive statistical analysis and fit the seasonal component. In the second we discuss the parameter estimation, while in the final part we implement the pricing method outlined above and analyze its sensitivities with respect to model and contract parameters. Partial results in subsections \ref{stanaly} and \ref{prinum} have been previously considered in Porthiyas (2019).
\subsection{Statistical analysis and parameter estimation}\label{stanaly}
Daily temperature data (in degree Celsius) at Toronto  from January 1st, 2013 to November 15th, 2018 have been collected from Environment and Climate Change, Canada. The data is gathered from the Pearson International Airport weather station and yield 2145 data points.
Observations consist of an average between the daily maximum and minimum temperatures. Missing observations are replaced by a seven-day moving average around the missing point.

\begin{figure}[htb!]
\centering
\includegraphics[width=\textwidth]{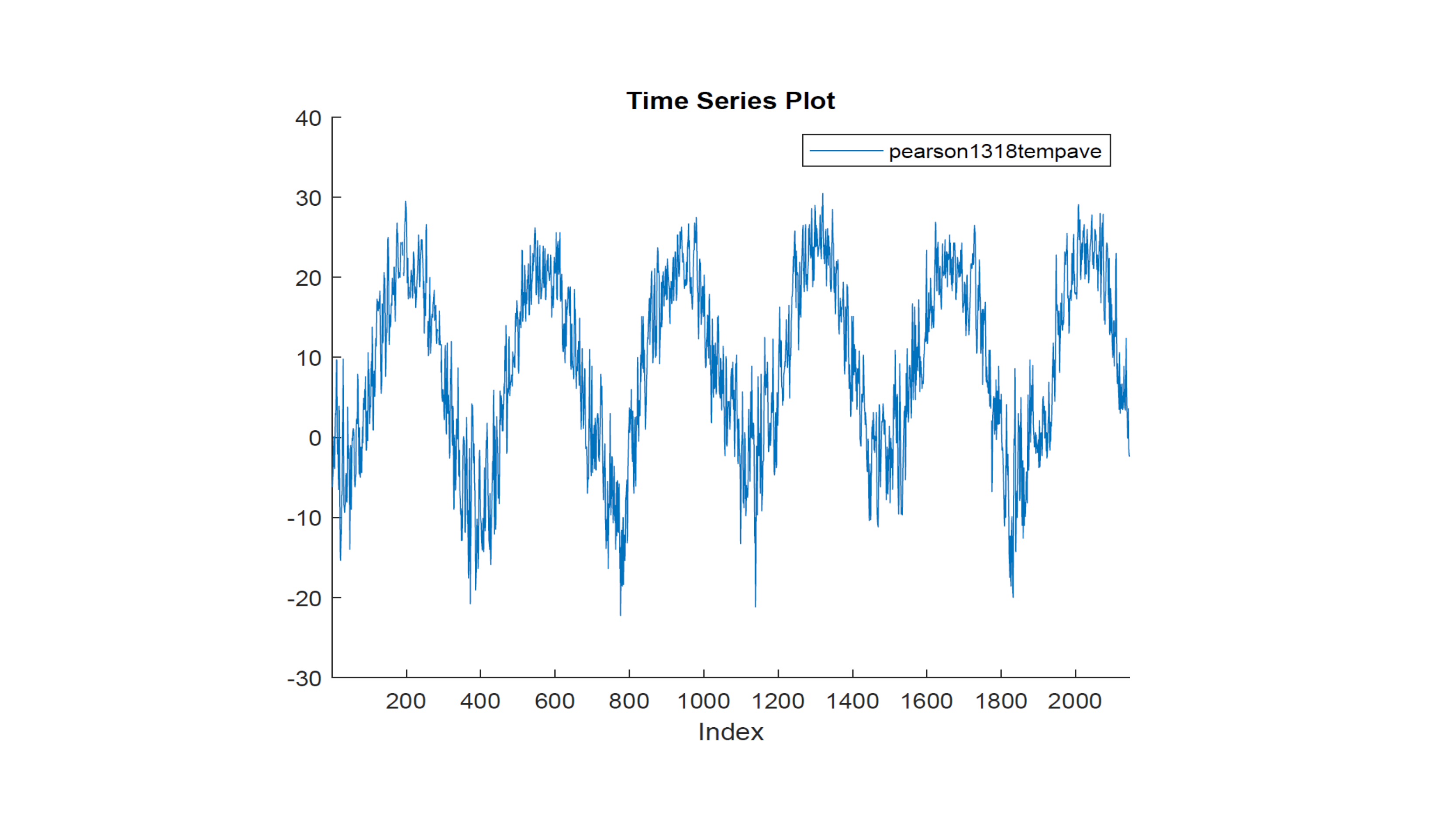}
\caption{Historic daily average temperature of Toronto from 1/1/2013 to 15/11/2018}
\label{fig:TOavgtemp}
\end{figure}

A preliminary statistical analysis of the temperature data shows the descriptive statistics as in Table \ref{tab:stats}. As can be seen, the skewness of the data is negative indicating a longer tail to the left. The kurtosis is less than 3 indicating more frequent but modest movements of temperature than would be expected under assumptions of normal distribution.

\begin{table}[htp!]
	\centering
	\caption{Summary of the series}
	\label{tab:stats}
	\begin{tabular}{ c | c | c | c | c | c}
 	\hline
 	Mean & \ Minimum & \ Maximum  & \ Std Dev  & \ Skewness   & \ Kurtosis \\
 	 		\hline
 	9.0483 & \ -22.30 & \ 30.45 & \ 11.0593 & \ -0.3021 & \ 2.1481 \\
 	 	
 	\hline
	\end{tabular}
\end{table}

It can be observed from both the histogram and the kernel density estimate in Figure \ref{fig:histksd} that the temperature data is bimodal. The left peak is centered around the mean temperature in winter and right peak is centered around the mean temperature in summer.\\

\begin{figure}[htb!]
\centering
\includegraphics[width=\textwidth]{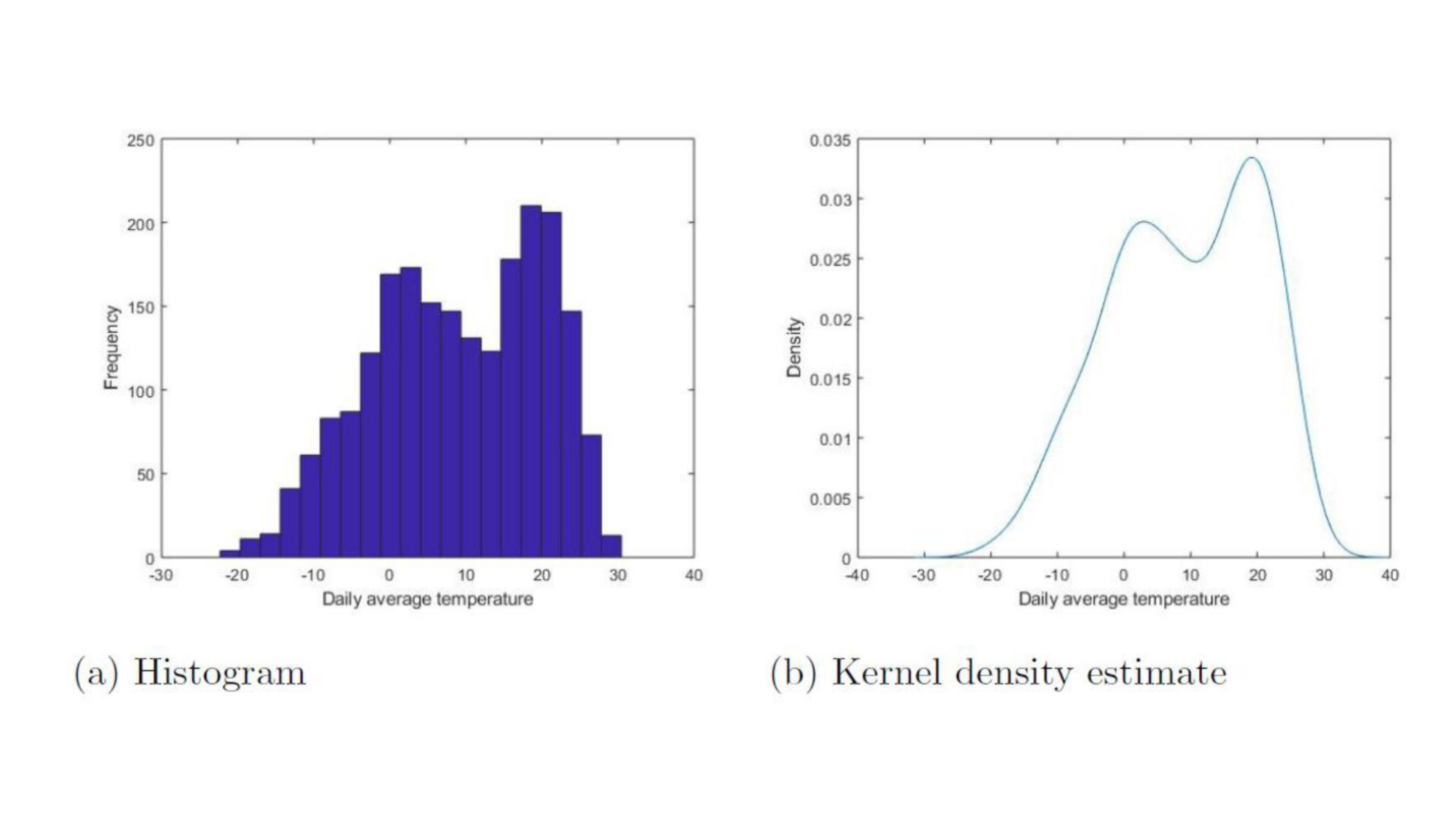}
\caption{Histogram and Kernel density estimate of daily average
temperatures in Toronto from 1/1/2013 to 15/11/2018}
\label{fig:histksd}
\end{figure}

\begin{table}
	\centering
	\caption{Kolmogorov-Smirnov test results}
	\label{tab:kstest}
	\begin{tabular}{ c  | c | c }
 	\hline
 	  p-value &  KSSTAT  & Critical value \\
 	\hline 	
 0 &  0.6889 &  0.0292 \\

 \hline

	\end{tabular}
\end{table}

Table \ref{tab:kstest} shows the results of a Kolmogorov–Smirnov test. This is a goodness-of-fit test to verify whether the data is from a normal distribution. It can be concluded from the p-value of zero and a KSSTAT value significantly greater than the critical value, that the temperature data do not seem to follow a normal distribution.

The seasonal component as  described in equation (\ref{eq:seass2}) is adjusted via a regression model. The results are shown in table \ref{tab:seass}.

\begin{table}
  \centering
 \begin{tabular}{|c|c|c|c|c|c|}
   \hline
   & Estimate&       SE &       t-Stat  & Conf. int.  &   pValue \\ \hline
    $ b_0 $&   7.9733 &     0.20221 &    39.431 & (7.573, 8.359) &  6.8857e-256 \\ \hline
    $b_1$  &       0.0008223   & 0.059639  &   5.0812 & (0.0005043, 0.00114) &   4.076e-07 \\ \hline
    $b_2$  &     -5.8796  &    0.14176    & -41.476  &  (-6.143, -5.590)& 3.103e-276 \\ \hline
    $b_3$ &      -12.866  &    0.14287  &  -90.052  &   (-13.13, -12.57)   &      0    \\    \hline
 \end{tabular}
  \caption{}\label{tab:seass}
\end{table}


\begin{figure}[htb!]
\centering
\includegraphics[width=\textwidth]{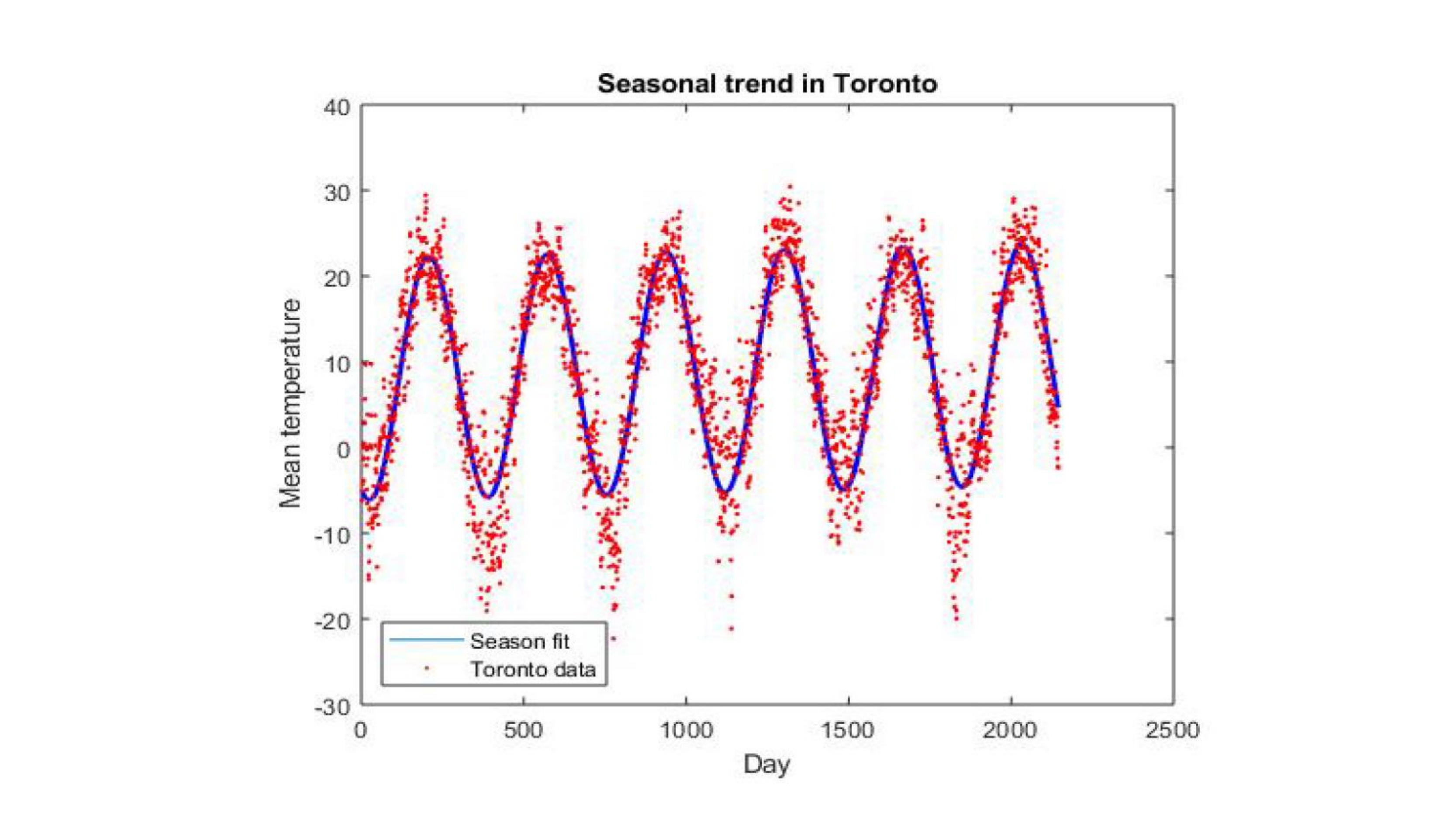}
\caption{Seasonal trend for Toronto daily mean temperature}
\label{fig:SeasFit}
\end{figure}

As it can be seen from Table \ref{tab:seass}, the slope term  $b_1$ in the regression fit is small but significantly different for zero,  which indicates the existence of a  linear trend in temperature rising, consistent with other climatic studies signaling the past decade as the warmest one since temperature is recorded. It must be noted in those cases, a larger set of temperature data for 40 years or more was used.
\subsection{Parameter estimation}
We base our analysis on the log-return series given by:
 \begin{equation}\label{eq:logretpri}
   X_{j \Delta}=\log \left( \frac{T_{(j+1) \Delta}}{T_{j \Delta}} \right)=Y_{(j+1) \Delta}-Y_{j \Delta}, \; j=1,2,\ldots,n
 \end{equation}
 where $\Delta>0$ is the frequency at which the data is registered, typically daily observations. Notice that the observations are independent but not equally distributed.\\
 We estimate the parameters in the model using a likelihood approach combined with the method of moments to set the initial estimate value. In addition, a method of minimum distance based on the characteristic function is considered.
 In Figure \ref{fig:simtraj} bottom, a simulated temperature graph for 2018 is shown, compared with the actual observations( top figure).
 \begin{figure}[htb!]
\centering
\includegraphics[width=\textwidth]{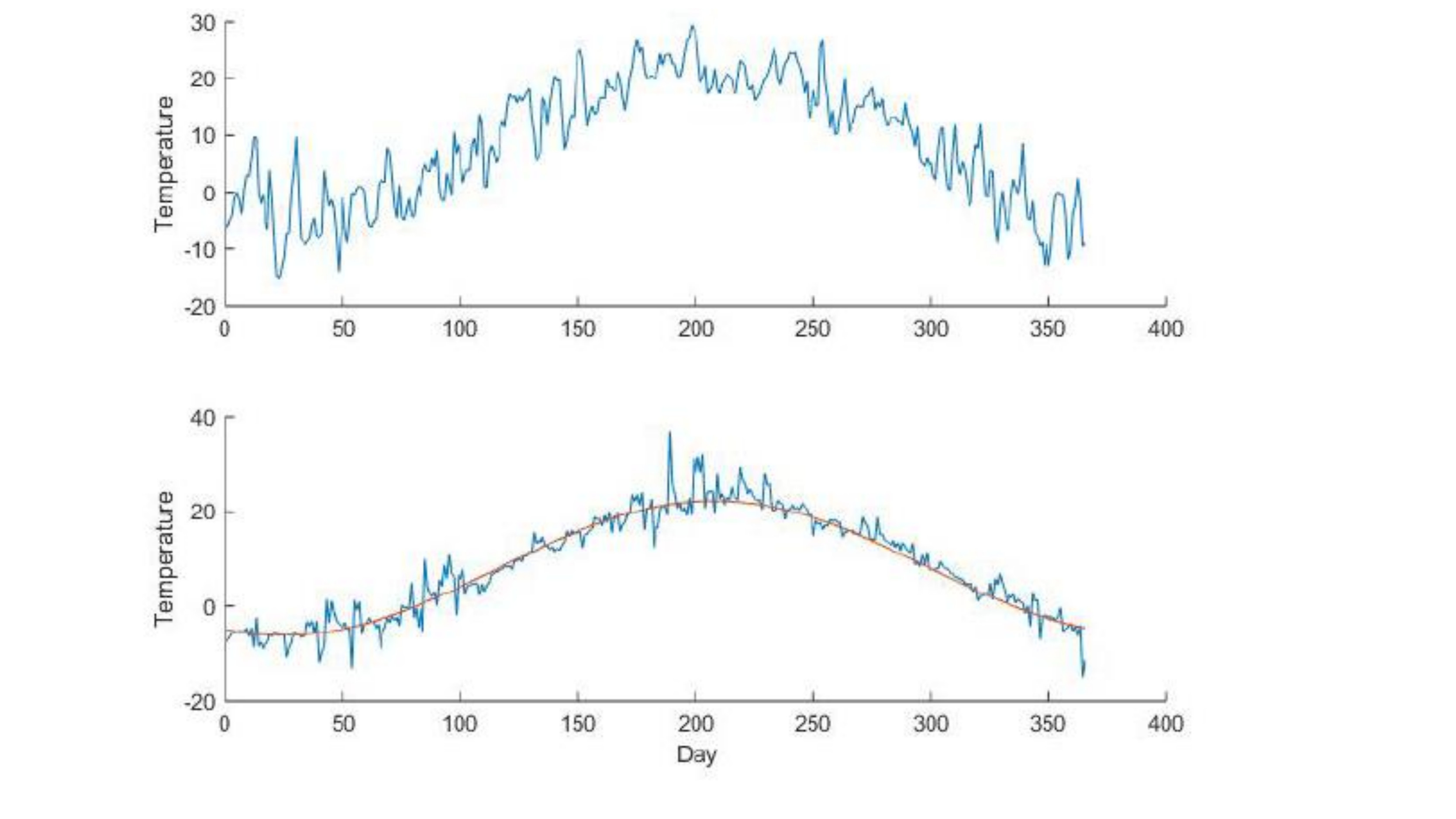}
\caption{Simulated temperature graph for 2018 is shown, compared with the actual observations}
\label{fig:simtraj}
\end{figure}
Figure \ref{fig:pricevsalpha} shows simulated trajectories for different values of the mean-reverting level (left) and how the price of the weather contract changes for different values of the same parameter.
\begin{figure}[htb!]
\centering
\includegraphics[width=\textwidth]{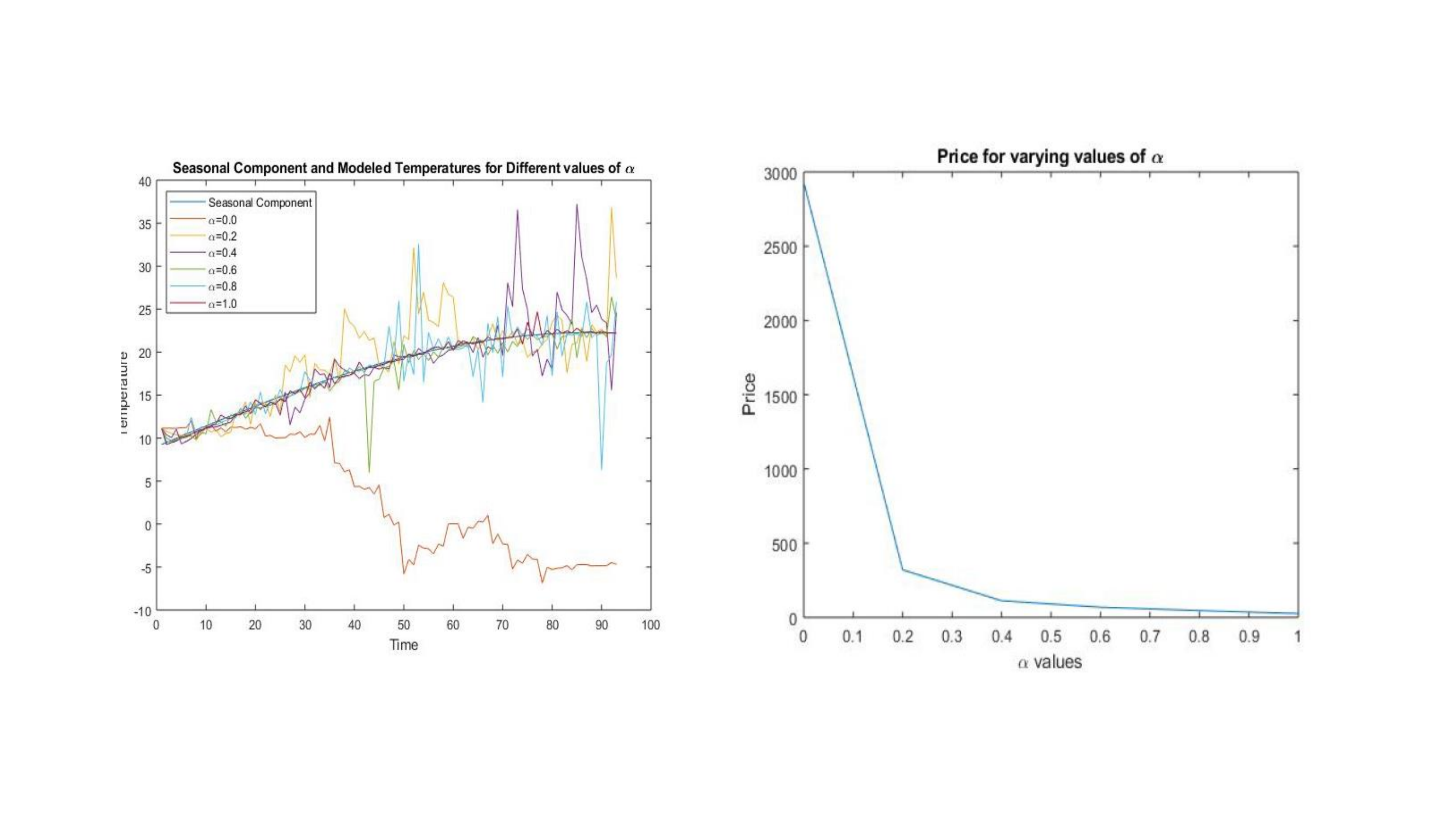}
\caption{Simulated trajectories for different values of the mean-reverting level}
\label{fig:pricevsalpha}
\end{figure}
\section{Acknowledgments}
The author would like to thank  the Natural Sciences and Engineering Research Council of Canada for its support.
\section{Conclusions}
A mean-reverting time-changed Levy process with periodic mean-reverting level and volatility offers a fair model for temperatures at Pearson International Airport temperatures. \\
On the other hand, pricing methods based on Fourier expansions provide an alternative algorithm under the models and the underlying series considered. Weather temperature prices are efficiently computed on a PC in reasonable time.

\end{document}